\documentclass[conference]{IEEEtran}
\IEEEoverridecommandlockouts
\usepackage{cite}
\usepackage{amsmath,amssymb,amsfonts}
\usepackage{algorithm}
\usepackage{algpseudocode}

\usepackage{graphicx}
\usepackage{epstopdf}
\usepackage{textcomp}
\usepackage{xcolor}
\usepackage{comment}
\usepackage{lipsum}
\usepackage{fullpage}

\usepackage{caption}
\usepackage{graphicx}
\usepackage{subfigure}
\usepackage{amssymb}
\usepackage{booktabs}
\usepackage{bm}
\usepackage[utf8]{inputenc}
\usepackage[english]{babel}
\usepackage{amsmath}
\usepackage{amsthm}
\usepackage{dsfont}
\newtheorem{theo.a}{Proposition}

\long\def\comment#1{}

\newcommand{\be}{\begin{equation}}
\newcommand{\ee}{\end{equation}}

\newtheorem{theorem}{Theorem}

\newfont{\bbb}{msbm10 scaled 700}

\newfont{\bb}{msbm10 scaled 1100}

\newcommand{\ev}{{\bf e}}

\newcommand{\xv}{{\bf x}}
\newcommand{\yv}{{\bf y}}

\newcommand{\Am}{{\bf A}}

\newcommand{\Hm}{{\bf H}}

\newcommand{\Sm}{{\bf S}}

\newcommand{\Wm}{{\bf W}}

\newcommand{\Ac}{{\cal A}}

\newcommand{\Ic}{{\cal I}}

\newcommand{\Kc}{{\cal K}}

\newcommand{\Nc}{{\cal N}}

\newcommand{\Pc}{{\cal P}}

\newcommand{\Sc}{{\cal S}}

\newcommand{\RNum}[1]{\uppercase\expandafter{\romannumeral #1\relax}}

\newcommand{\Sigmam}{\hbox{\boldmath$\Sigma$}}

\newcommand{\eqdef}{\stackrel{\Delta}{=}}

\begin{document}

\title{Information Theoretic Data Injection Attacks with Sparsity Constraints \\
\thanks{This research was supported in part by the European Commission through the H2020-MSCA-RISE-2019 program under grant 872172 and in part by the China Scholarship Council.}
}

\author{\IEEEauthorblockN{Xiuzhen Ye$^*$, I\~naki Esnaola$^{*\dag}$, Samir M. Perlaza $^{\S\dag}$, and Robert F. Harrison$^{*}$}

\IEEEauthorblockA{$^*$Dept. of Automatic  Control and Systems Engineering, University of Sheffield, Sheffield S1 3JD, UK\\
 $^\dag$Dept. of Electrical Engineering, Princeton University, Princeton, NJ 08544, USA\\
 $^\S$INRIA, Centre de Recherche de Sophia Antipolis - Méditerranée, France\\
}
}
\maketitle
\begin{abstract}
Information theoretic sparse attacks that minimize simultaneously the information obtained by the operator and the probability of detection are studied in a Bayesian state estimation setting. The attack construction is formulated as an optimization problem that aims to minimize the mutual information between the state variables and the observations while guaranteeing the stealth of the attack. Stealth is described in terms of the Kullback-Leibler (KL) divergence between the distributions of the observations under attack and without attack. To overcome the difficulty posed by the combinatorial nature of a sparse attack construction, the attack case in which only one sensor is compromised is analytically solved first. The insight generated in this case is then used to propose a greedy algorithm that constructs random sparse attacks. The performance of the proposed attack is evaluated in the IEEE 30 Bus Test Case.
\end{abstract}
\section{Introduction}
State estimation enables efficient, scalable, and secure operation of power systems. 
This is in part thanks to monitoring and control processes that are supported by Supervisory Control and Data Acquisition (SCADA) systems and more recently by advanced communication systems that acquire and transmit observations to a state estimator~\cite{AA_PSstateestimation_04}.
This cyber layer exposes the system to malicious attacks that exploit the vulnerabilities of the sensing and communication infrastructure. One of the main threats faced by modern power systems are data injection attacks (DIAs)~\cite{LY_TISSEC_11} {that} alter the state estimate of the operator by compromising the system observations.
A large body of literature is concerned with the case in which attack detection is performed by a residual test~\cite{VO_SGC_22} under the assumption that state estimation is deterministic. In this setting, constructing DIAs that require access to a small set of observations {yields} optimization problems with sparsity constraints, which are often difficult to solve. In~\cite{KP_TSG_11}, it is shown that the operator can secure a small fraction of observations to make the attack construction significantly harder. This problem has been studied extensively in the literature in both centralized and decentralized scenarios~\cite{TA_SGC_11},~ \cite{CKKPT_SPM_12},~\cite{MO_JSAC_13},~\cite{EPP_gsip_14}.

The unprecedented data acquisition capabilities in the smart grid {elevate} the threat of attack precisely because accurate stochastic models can be generated for the system. In view of this, attack constructions that exploit this prior knowledge can be posed within a Bayesian framework~\cite{OK_TSG_11}. In this setting, the attack detection problem is no longer cast as a residual test. Instead, detection strategies consider the likelihood ratio test~\cite{IE_TSG_16} or alternatively machine learning methods~\cite{OM_TNNLS_16}.
The operator produces a stochastic model of the system based on the observations generated by the monitoring system.
Moreover, data analytics on the system depend on the reliability of the observations that are used with a variety of estimation, statistical and machine learning tools that provide the operator with different insight.
In view of this, it is essential to assess attacks in fundamental terms to understand the impact over a wide range of estimation and data analysis paradigms.

Information theoretic attacks are first introduced in~\cite{KS_SGC_17} and then generalized in~\cite{KS_TSG_19}. In this approach, attack disruption is measured in terms of two information measures: (a) the mutual information between the state variables and the observations under attack; and (b) the probability of detection, which is governed by the Kullback-Leibler (KL) divergence. The advantage of using these information measures is that the attack disrupts a wide range of estimation, statistical and machine learning methods that are available to the operator.
Given that the attack vector corrupts the observations in an additive fashion, mutual information minimization yields a Gaussian attack construction that has the maximum entropy, i.e. maximum uncertainty, among all the distributions with fixed variance~\cite{book_EIT}. From a practical point of view, the assumption is {validated} given the data shared by Electricity North West Limited~\!\!{\cite{GE_TSG_18}}.
In this case, mutual information decreases monotonically with the variance of the attack vector entries~\cite{SI_TIT_13} and the converse holds for the probability of attack detection.
The information theoretic attacks in~\cite{KS_TSG_19} require that the attacker {tampers} with all the observations used by the operator~\cite{SE_ICASSP_19}. Hence, incorporating sparsity constraints with information theoretic attacks is still an open problem that requires novel approaches. In this paper, we present a novel information theoretic sparse attack construction based on a greedy observation selection mechanism.

A brief description of notation follows. Consider matrix $\Am\in\mathds{R}^{m\times n}$, then $(\Am)_{ij}$ denotes the entry in row $i$ and column $j$. We denote by $\Am_{\Ic}$ the matrix formed with the rows of $\Am\in\mathds{R}^{m\times n}$ given by the indices in $\Ic\subseteq\{1, \ldots, m\}$ {in increasing order}. We denote the complement of set $\Ic$ by ${\Ic^{\sf{c}}}$. The elementary vector $\ev_i$ is a vector of zeros with a one in {the entry $i$}. Random variables are denoted by capital letters and their realizations by the correponding lower case, e.g. $x$ is a realization of the random variable $X$. Vectors of $n$ random variables are denoted by a superscript, e.g. $X^n=(X_1, \ldots, X_n)$ with corresponding realizations denoted by $\xv$. The set of positive semidefinite matrices of size $n\times n$ is denoted by $S_{+}^n$.
\section{System model}\label{system model}
\subsection{Power system state estimation}
In a power system the state vector $\xv \in{\mathds{R}^n}$ containing the voltages and angles at all the generation and load buses describes the operation state of the system. State vector $\xv$ is observed by the acquisition function $F: {\mathds{R}^n} \rightarrow {\mathds{R}^m}$. A linearized observation model is considered for state estimation, yielding the observation model
\begin{equation}\label{eq:obs_noattack}
  Y^m  = \textbf{H}\xv+Z^m,
\end{equation}
where $\textbf{H} \in {\mathds{R}^{m \times n}}$ is the Jacobian of {the function $F$ at a given operating point and is determined by the system components and the topology of the network}. The vector containing observations $Y^m$ is corrupted by {additive white Gaussian noise} introduced by the sensors~{\cite{AA_PSstateestimation_04},~\cite{GJ_PSanalysis_1994}}. The noise vector $Z^m$ follows a multivariate Gaussian distribution $Z^m \sim \mathcal{N}(\textbf{0},\sigma^2 \textrm{\textbf{I}}_m)$, where $\sigma^2$ is the noise variance.

In a Bayesian estimation framework, the state variables are described by a vector of random variables $X^n$ with a given distribution. As the Gaussian distribution has the maximum entropy among all distributions with the same variance, we assume $X^n$ follows a multivariable Gaussian distribution with zero mean and covariance matrix $\bm{\Sigma}_{X\!X} \in S_{+}^n$. From~(\ref{eq:obs_noattack}), it follows that the vector of observations is zero mean and with covariance matrix
\begin{equation}\label{2} 
  \bm{\Sigma}_{Y\!Y} = \textbf{H}\bm{\Sigma}_{X\!X}\textbf{H}^{\sf{T}}+\sigma^2 \textrm{\textbf{I}}_m.
\end{equation}

{The resulting observations are corrupted by the malicious attack vector
\begin{equation}\label{3}
  A^m \sim P_{A^m},
\end{equation}
where $P_{A^m}$ is the distribution of the random attack vector $A^m$}.
Since the Gaussian distribution minimizes the mutual information between the state variables and the compromised observations {with} a fixed covariance matrix~\cite{SI_TIT_13}, we adopt a Gaussian random attack {framework} given by
\begin{equation}\label{eq:Gauss_attack}
  A^m \sim \mathcal{N} (\textbf{0}, \bm{\Sigma}_{A\!A}),
\end{equation}
where $\bm{\Sigma}_{A\!A}$ is the covariance matrix of attack vector $A^m$. Consequently, the compromised observations denoted by $Y_A^m$ are given by
\begin{equation}
\label{eq:obs_attack}
  Y_A^m  = \textbf{H}X^n+Z^m + A^m,
\end{equation}
where $Y_A^m$ follows a multivariate Gaussian distribution given by
\begin{equation}\label{6}
  Y^m_A  \sim \mathcal{N} (\textbf{0},\bm{\Sigma}_{Y_A\!Y_A})
\end{equation}
with $\bm{\Sigma}_{Y_A\!Y_A} = \textbf{H}\Sigmam_{X\!X}\textbf{H}^{\sf{T}} + \sigma^2 \textrm{\textbf{I}}_m + \bm{\Sigma}_{A\!A}$.
\subsection{Attack Detection}
As a part of a security strategy, the operator implements an attack detection procedure prior to performing state estimation. Detection is cast as a hypothesis testing problem given by:
\begin{align}\label{eq:hypoth_attack}
  \mathcal{H}_0&:\textrm{There is no attack,}\\ 
  \mathcal{H}_1&:\textrm{Observations are compromised}.
\end{align}
In this setting, the optimal test is the likelihood ratio test (LRT)~\cite{JN_LRT_33} given by
\begin{equation}\label{lrt}
L(\textbf{\textrm{y}}) = \frac{f_{Y_A^m}(\textbf{\textrm{y}})}{f_{Y^m}(\textbf{\textrm{y}})} \overset{{\cal H}_1}{\underset{{\cal H}_0}{\gtrless}} \tau,
\end{equation}
where $\textbf{\textrm{y}}$ is the realization of the observations to be tested for attack; $f_{Y_A^m}(\textbf{\textrm{y}})$ is the probability density function (pdf) of $Y_A^m$ in (\ref{eq:obs_attack}), $f_{Y^m}(\yv)$ is the pdf of $Y^m$ in~(\ref{eq:obs_noattack}), and $\tau\in\mathds{R}_+$ in~(\ref{lrt}) is the decision threshold.
The performance of the test is assessed in terms of the Type I error, defined as $\alpha\eqdef \mathds{P}\left[L(\bar{Y}^m)\geq\tau\right]$ with $\bar{Y}^m\thicksim P_{Y^m}$, and the Type II error, denoted by $\beta\eqdef \mathds{P}\left[L(\bar{Y}^m)<\tau\right]$ with $\bar{Y}^m\thicksim P_{Y_A^m}$.
Note that the LRT is optimal, and therefore, changing the value of $\tau$ is equivalent to changing the tradeoff between Type I and Type II errors.

\section{Sparse {information theoretic} attacks}\label{sparse construction}
\subsection{{Information theoretic} setting}
The attack construction in~\cite{KS_TSG_19} incorporates a detection constraint based on the KL divergence between the distributions $P_{Y^m_A}$ in~(\ref{eq:obs_attack}) and $P_{Y^m}$ in~(\ref{eq:obs_noattack}) which results in the construction of {\it stealth attacks}. Specifically, the construction is given by the solution to the following optimization problem:
\begin{equation}\label{eq:stealth_opt}
  \min_{P_{A^m}} I(X^n;Y^m_A)+ \lambda D(P_{Y_A^m}\|P_{Y^m}),
\end{equation}
where $I(X;Y)$ is the mutual information between random variables $X$ and $Y$, $D(P\| Q)$ denotes the KL divergence between distributions $P$ and $Q$, and $\lambda\geq 1$ is the weighting parameter that determines the tradeoff between attack disruption and probability of detection. Note that the optimization in~\eqref{eq:stealth_opt} searches for the distribution of the attack vector of random variables over the set of Gaussian multivariate distributions of $m$ dimensions, or equivalently, it chooses the optimal covariance matrix for the distribution of the attack. It is shown in \cite{KS_TSG_19} that the optimal Gaussian attack is given by $\bar{P}_{A^m}=\Nc(\mathbf{0},\bar{\Sigmam})$ where
\be\label{eq:stealth_cov}
\bar{\Sigmam}={\lambda^{-1}}\Hm\Sigmam_{X\! X}\Hm^{\sf T}.
\ee
Note that in~\cite{KS_TSG_19}, the construction of the stealth attack vector is not sparse, indeed all the components of the attack realizations are nonzero with probability one, i.e. $\mathds{P}\left[|\textnormal{supp}({A^m})|=m\right]=1$. We define the support of the attack vector ${A^m}$ by
\be
\textnormal{supp}({A^m})\eqdef\left\{i:\mathds{P}\left[A_i=0\right]=0\right\}.
\ee
\subsection{Sparse attack formulation}
Given that the operator is likely to have access control policies in place~\cite{Scada_book}
, an attack construction that requires access to all the observations is costly and unrealistic for the attacker in most scenarios. For that reason, in the following we study stealth attack constructions that require access to a limited number of sensors. In particular, we pose the optimization problem with sparsity constraints by considering distributions over the attack vector that put non-zero mass on at most $k\leq m$ {attack vector components}.
Thus, we include the additional requirement that $\left| \textnormal{supp}({A^m})\right|=k$ in the attack construction. In view of this, the attacker chooses the distribution of the attack vector over the set of multivariate Gaussian distributions given by
\be
\Pc_k\eqdef\left \{P_{A^m}:\left| \textnormal{supp}({A^m})\right | =k\right \}.
\ee
The resulting $k$-sparse stealth attack construction is therefore posed as the optimization problem:
\be\label{eq:k_sparse_stealth_opt}
  \min_{P_{A^m}\in\Pc_k} I(X^n;Y^m_A)+ \lambda D(P_{Y_A^m}\|P_{Y^m}).
\ee
Solving this problem is hard in general owing to the combinatorial nature of the attack vector support selection. For that reason, in Section \ref{optimal sparse construction} we tackle the problem by  proposing a greedy attack construction algorithm that results in $k$-sparse attack vectors.
\subsection{Gaussian sparse attack construction}

In the following, we particularize the attack construction in~\eqref{eq:stealth_opt} by considering Gaussian distributed state variables, i.e. $X^n\thicksim\Nc(\mathbf{0},\Sigmam_{X\!X})$, and assuming that the attack vector follows the Gaussian distribution given in~\eqref{eq:Gauss_attack}. In this setting, the optimization problem in~\eqref{eq:stealth_opt} is equivalent~\cite{KS_TSG_19} to the following optimization problem:
\be
\begin{aligned}\label{eq:Gaussian_stealth_constr}
\min_{\bm{\Sigma}_{A\!A} \in S_{+}^m } & (1-\lambda) \log |\textbf{I}_m + \Wm\bm{\Sigma}_{A\!A}| \\
&- \log | \sigma^2\textbf{I}_m + \bm{\Sigma}_{A\!A}|  + \lambda \textrm{tr}(\Wm\bm{\Sigma}_{A\!A})),
\end{aligned}
\ee
where $\Wm\eqdef\Sigmam^{-1}_{Y\!Y}$. In order to incorporate sparsity constraints in~\eqref{eq:Gaussian_stealth_constr}, the minimization domain is constrained to the set of covariance matrices that induce $k$-sparse supports over the attack vectors, i.e., the set given by
\be
\label{eq:cov_sparse}
\Sc_k\eqdef \left\{\Sm\in S_{+}^m: \| \textnormal{diag}(\Sm) \|_0=k \right\},
\ee
where $\textnormal{diag}(\Sm)$ denotes the vector formed by the diagonal entries of $\Sm$.
Solving~\eqref{eq:Gaussian_stealth_constr} within the optimization domain specified by~\eqref{eq:cov_sparse} re-casts the equivalent $k$-sparse stealth attack construction problem {in~\eqref{eq:k_sparse_stealth_opt} as follows}:
\be
\begin{aligned}\label{eq:Gaussian_k_stealth_constr}
\min_{\bm{\Sigma}_{A\!A} \in \Sc_k } & (1-\lambda) \log |\textbf{I}_m + \Wm\bm{\Sigma}_{A\!A}| \\
&- \log | \sigma^2\textbf{I}_m + \bm{\Sigma}_{A\!A}|  + \lambda \textrm{tr}(\Wm\bm{\Sigma}_{A\!A})).
\end{aligned}
\ee
\subsection{Optimal single observation attack case}\label{sec:single}
Despite having narrowed it down to Gaussian distributions, the above optimization problem is still {challenging} and combinatorial in nature. For that reason, we first tackle the case in which the attacker only comprises one sensor, i.e. $k=1$.  The rationale for this is that we use the insight developed for the single sensor case in the construction of the general $k$-sparse case. The following theorem provides the optimal solution for the case in which the attacker corrupts a single sensor.
\begin{theorem}\label{Theorem1}
The solution to the sparse stealth attack construction problem in~(\ref{eq:Gaussian_k_stealth_constr}) for the case $k=1$ is given by
\begin{equation}\label{26}
\bar{\Sigmam}_{A\!A}=\bar{\sigma}^2\ev_{\alpha}\ev_{\alpha}^{\sf T},
\end{equation}
where
\begin{IEEEeqnarray}{ll}
\label{eq:alpha}
\alpha&=\textnormal{arg}\min_{i} \left\{(\Wm)_{ii}\right\},\\
\label{eq:sigma_opt}
\bar{\sigma}^2&=- \frac{\sigma^2}{2} + \frac{1}{2}\left( \sigma^4 -\frac{4 (\underbar{$w$}\sigma^2 -1)}{\lambda \underbar{$w$}^2} \right)^{\frac{1}{2}},
\end{IEEEeqnarray}
with $\underbar{$w$}\eqdef (\Wm)_{\alpha\alpha}$.
\end{theorem}
\begin{proof}
We start by noting that for $k=1$ the set of attack covariance matrices is given by
\be
\Sc_1\eqdef\!\!\!\!\bigcup_{i=1,\ldots, m}\!\!\!\!\left\{\Sm\in S_{+}^m\!: \!\Sm=\sigma^2_i \ev_i\ev_i^{\sf T} \textnormal{with}\; \sigma_i\!\in\!\mathds{R}_+\right\}.
\ee
The covariance matrices in set $\Sc_1$ comprise matrices with a single nonzero element in the diagonal. The non-zero entry $i$ denotes the index of the sensor that is attacked.
Let $i \in \{1,2,...,m\}$ be the index of the non-zero entry of the covariance matrix $\bar{\Sigmam}_{A\!A}$. The non-zero entry denoted by $\sigma_i^2$ is the variance of the random variable used to attack observation $i$.

Let $\lambda >1$ and restrict the optimization domain in~(\ref{eq:Gaussian_k_stealth_constr}) to $\Sc_1$. {Thus, the resulting optimization problem is equivalent to:}
\begin{equation}\label{eq:opt_k1}
 {\min_{\bar{\sigma}>0}} \min_{i} \log \frac{(1+(\Wm)_{ii}\bar{\sigma}^2)^{ 1-\lambda } }{(\sigma^2 + \bar{\sigma}^2)}  + \lambda (\Wm)_{ii} \bar{\sigma}^2.
\end{equation}
We proceed by solving the inner part of the optimization problem above. Consider the cost given by
{
\begin{equation}\label{inner_opt}
f((\Wm)_{ii}) \eqdef \log \frac{(1+(\Wm)_{ii}\bar{\sigma}^2)^{ 1-\lambda } }{(\sigma^2 + \bar{\sigma}^2)} \! + \!
\lambda (\Wm)_{ii} \bar{\sigma}^2,
\end{equation}}
which can be rewritten as
\begin{equation}\label{eq:ft}
f(t)=(1-\lambda)\log t - \log (\sigma^2 + \bar{\sigma}^2) + \lambda t - \lambda,
\end{equation}
where $t=1+(\Wm)_{ii}\bar{\sigma}^2$. It follows that~\eqref{eq:ft} is convex with respect to $t$ because $\lambda t$ is a linear term and $(1-\lambda)\log t$ is convex in $t$ for $\lambda > 1$. Therefore, $f((\Wm)_{ii})$ is convex with respect to $(\Wm)_{ii}$ and the minimum is attained for $(\Wm)_{ii} = - \frac{1}{\lambda \bar{\sigma}^2 }$. Since $(\Wm)_{ii}>0$ the inner minimization in~\eqref{eq:opt_k1} is equivalent to selecting the index $i$ that minimizes $(\Wm)_{ii}$. The definition of $\alpha$ in~{\eqref{eq:alpha}} and $\underbar{$w$}$ {in~\eqref{eq:sigma_opt}} follow from this observation.

\begin{algorithm}
  \caption{$k$-sparse stealth attack construction}
  \label{alg:greedy}
  \begin{algorithmic}[1]
    \Require the observation matrix $\Hm$; the covariance matrix of the state variables $\Sigmam_{X\!X}$; the variance of the noise $\sigma^2$; and the weighting parameter $\lambda$; {number of nonzero attack vector components {$k$}}.
    \Ensure the covariance matrix of the attack vector $\bar{\Sigmam}_{AA}$; and the set of indices of attacked sensors $\Ac$.
    \State Set $\Ac_0=\left\{\emptyset \right\}$
    \For {$j=1$ to $k$}
    \State Set $\Hm_j={\Hm_{\Ac^{\sf c}_{j-1}}}$
    \State Compute $\Wm_j=\left(\Hm_j\Sigmam_{X\!X}\Hm_j^{\sf T}+\sigma^2\mathbf{I}_{|{\Ac^{\sf c}_{j-1}}|}\right)^{-1}$
    \State Set $\alpha_j=\textnormal{arg}\min_{i} \left\{(\Wm_j)_{ii}\right\}$,
    \State Set {$\underline{w}_{j} \eqdef (\Wm_j)_{\alpha_j\alpha_j}$}
    \State Set $\bar{\sigma}_j^2=- \frac{\sigma^2}{2} + \frac{1}{2}\left( \sigma^4 -\frac{4 (\underbar{$w$}_j\sigma^2 -1)}{\lambda \underbar{$w$}_j^2} \right)^{\frac{1}{2}}$
    \State Set $\Ac_j = \Ac_{j-1}\cup \left\{\alpha_j\right\} $
\EndFor
 \State Set $\Ac = \Ac_k$
\State Set $\bar{\Sigmam}_{A\!A} = \sum_{i\in\Ac}\bar{\sigma}_i^2\ev_i\ev_i^{\sf T}$
  \end{algorithmic}
\end{algorithm}
 
We now proceed to solve the outer optimization. In this case, the cost is given by
\begin{equation}\label{eq:gt}
g(r) = (1-\lambda)\log (1+\underbar{$w$}r) - \log (\sigma^2 + r) +\lambda\underbar{$w$} r,
\end{equation}
where {$r \eqdef \bar{\sigma}^2$}. Noticing that the above function has a single minimizer given by
\be
r=- \frac{\sigma^2}{2} + \frac{1}{2}\left( \sigma^4 -\frac{4 (\underbar{$w$}\sigma^2 -1)}{\lambda \underbar{$w$}^2} \right)^{\frac{1}{2}}
\ee
completes the proof.
\end{proof}

\section{Greedy construction of sparse attacks}\label{optimal sparse construction}
The extension to the $k$-sparse case of the solution proposed in Section~\ref{sec:single} does not get around the combinatorial optimization in~\eqref{eq:Gaussian_k_stealth_constr}. For that reason, in the following we propose a greedy construction that leverages the insight distilled in the $k=1$ case to select the set of $k$ attacked sensors. The construction is based on a classical greedy procedure that sequentially selects an observation to attack by maximizing the performance in terms of the decision at each step.
Let us denote by $\Ac$ the set of observation indices that are attacked, i.e. $\Ac\eqdef\textnormal{supp}({A^m})$.  The greedy algorithm operates by sequentially updating the entries in $\Ac$ by adding a new index in each step until  $k$ indices are selected. For that reason, the resulting entries of the attack vector are independent, and therefore, the covariance matrix of the attack vector obtained via the proposed greedy approach belongs to the set
\vspace{-2mm}
\be
\tilde{\Sc}_k\!\eqdef \bigcup_{\Kc}\!\left\{\Sm\!\in\! S_{+}^m\!\!:\Sm\!=\!\!\sum_{i\in\Kc}\sigma_i^2\ev_i\ev_i^{\sf T}\textnormal{with}\;\sigma_i\!\in\!\mathds{R}_+\!\right\},
\ee
where the union is over all subsets $\Kc\subseteq\left\{1,2, \ldots, m\right\}$ with $|\Kc|=k\leq m$. The proposed greedy construction is described in Algorithm \ref{alg:greedy}.
\begin{figure}[t!]
\centering
\includegraphics[width=8cm]{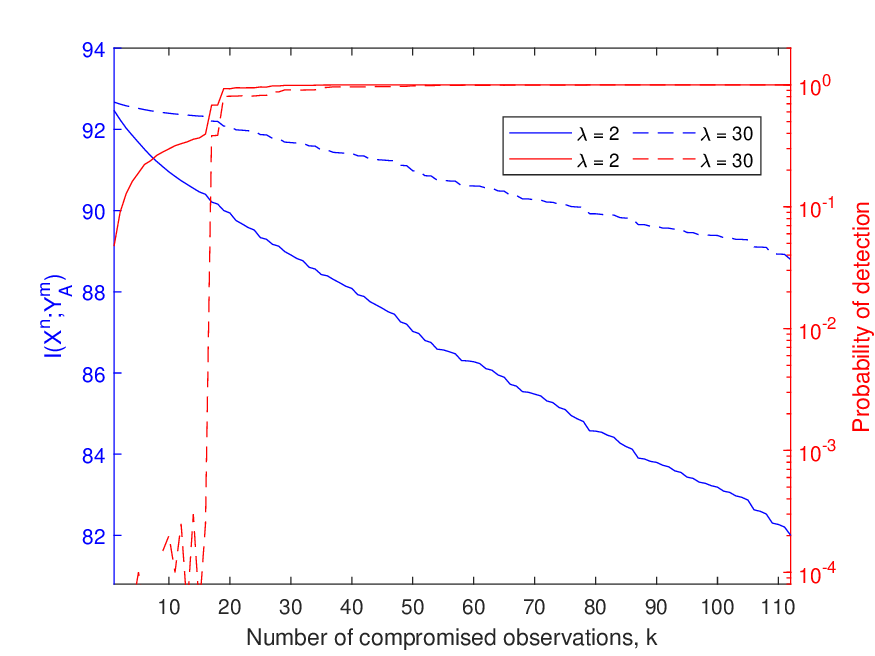}
\caption{Performance of the sparse attack in terms of mutual information, probability of detection for different values of $\lambda$ when $\textrm{SNR }=30 \textrm{dB}, \rho = 0.1, \tau = 2$ on the IEEE 30 Bus Test Case.}\label{fig:MI_prob}
\vspace{-5mm}
\end{figure}

\section{Numerical results}\label{simulation}

In this section, we present the simulation results on a direct current (DC) state estimation setting for the IEEE 30 Bus Test Case~\cite{UoW_ITC_99}. The voltage magnitudes are set to $1.0$ per unit, which implies that the state estimation is based on the observations of active power flow injections to all the buses and the active power flow between physically connected buses. The Jacobian matrix $\Hm$ is determined by the reactances of the branches and the topology of the system. MATPOWER~\cite{matpower} is adopted to generate $\Hm$.
To capture the statistical dependence between the state variables we adopt a Toeplitz model for the covariance matrix $\Sigmam_{X\!X}$ that arises in a wide range of practical settings, such as autoregressive stationary processes.
Specifically, we model the correlation between state variables $X_i$ and $X_j$ with the exponential decay parameter $\rho$ that results in $\left(\Sigmam_{X\!X}\right)_{ij} =  \rho^{|i-j|}$ with $i,j = 1,2,\ldots,n$.

In this setting, the performance of the proposed sparse stealth attack is a function of the correlation parameter $\rho$, noise variance $\sigma^2$, and the topology of the system as described by $\Hm$.
We define the signal to noise ratio (SNR) as
\be
\vspace{-1mm}
\textrm{SNR} \eqdef 10\log_{10}\left(\frac{\textrm{tr}(\textbf{H}\Sigmam_{\textrm{XX}}\textbf{H}^\textrm{\sf T})}{m\sigma^{2}}\right).
\ee
The results in this section are obtained by averaging \mbox{$2\times10^4$} realizations of the observations as described in~\eqref{eq:obs_attack}.
Fig.~\ref{fig:MI_prob} depicts the mutual information and the probability of detection that the attack constructed via Algorithm~\ref{alg:greedy} induces for different values of the number of compromised observations and the weighting parameter $\lambda$. As expected, the mutual information decreases monotonically, approximately linearly with the number of compromised observations, while the probability of detection increases monotonically. Interestingly, the probability of detection exhibits an abrupt increase that suggests a {\it threshold effect} when a critical number of compromised observations is reached. The weighting parameter $\lambda$ governs the minimum achievable probability of detection, e.g. a probability of detection of $10^{-2}$ is not attainable when $\lambda=2$. Indeed, increasing the value of $\lambda$ to $30$ yields a smaller probability of detection for small values of $k$ but the threshold effect takes place for the same number of compromised observations, for both values of $\lambda$. This suggests that the topology of the system governs the position of the threshold.

The variance of the random variables used to attack each sensor, the probability of detection, and the probability of false alarm as a function of the number of compromised observations are illustrated in Fig.~\ref{fig:prob_var_l2} and Fig.~\ref{fig:prob_var_l30} for $\lambda=2$ and $\lambda=30$, respectively.
As shown in Theorem 1, $\lambda$ is a scaling factor on the variances of the attack vector, and therefore, the values of the variance for the case $\lambda=2$ are simply scaled in the case $\lambda=30$.
\begin{figure*}[htbp]
\centering
\begin{minipage}[t]{0.49\textwidth}
\includegraphics[width=8cm]{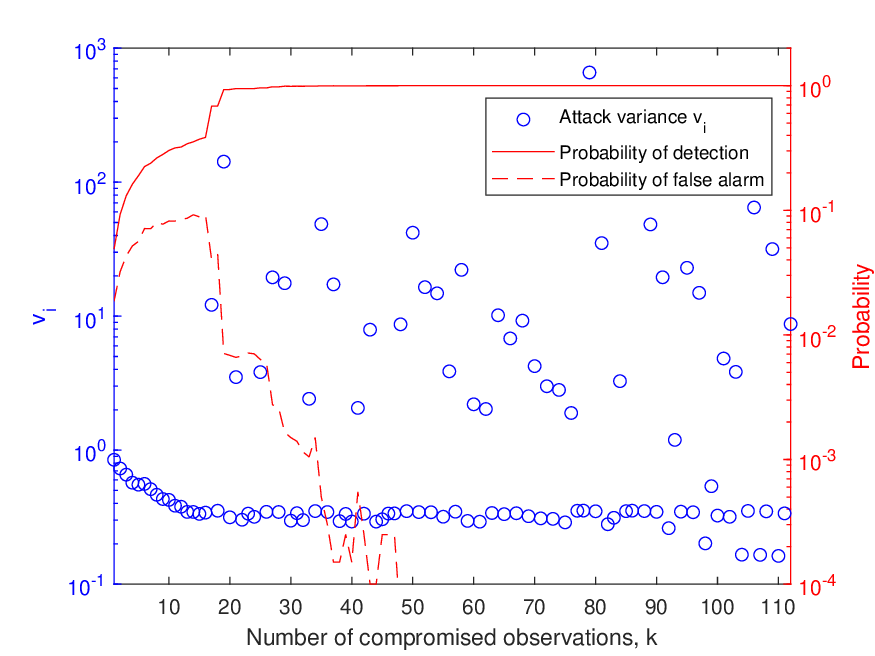}
\caption{Variance of the attack vector entries, probability of detection, and probability of false alarm of the sparse attack when $\lambda = 2, \textrm{SNR}=30\; \textrm{dB}, \rho = 0.1, \tau = 2$ on the IEEE 30 Bus Test Case.}\label{fig:prob_var_l2}
\end{minipage}
\centering
\begin{minipage}[t]{0.49\textwidth}
\includegraphics[width=8cm]{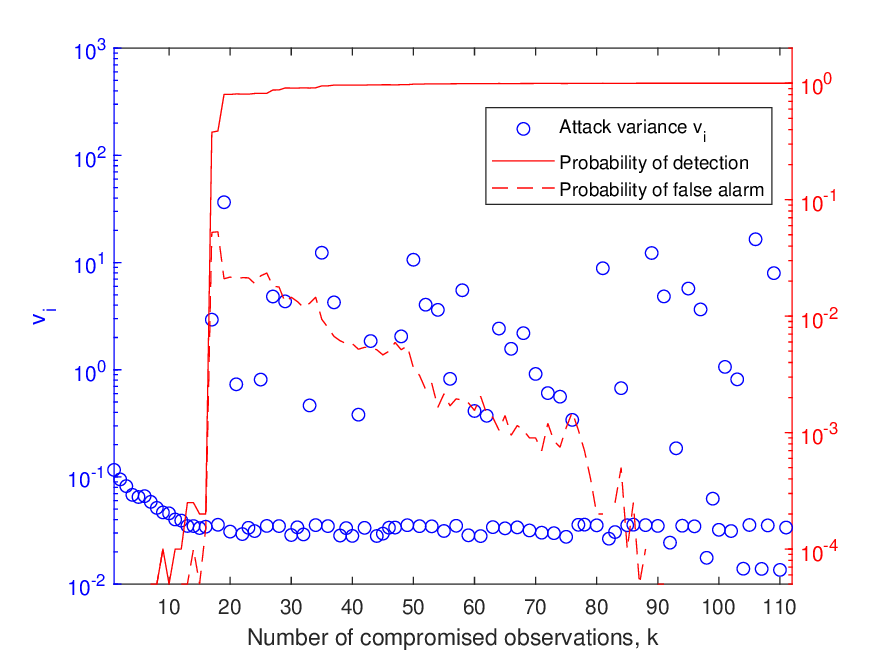}
\caption{Variance of the attack vector entries, probability of detection, and probability of false alarm of the sparse attack when $\lambda = 30, \textrm{SNR}=30\; \textrm{dB}, \rho = 0.1, \tau = 2$ on the IEEE 30 Bus Test Case.}\label{fig:prob_var_l30}
\end{minipage}
\vspace{-5mm}
\end{figure*}
There are two distinguishable attack regimes depending on the variance of the attack vector entries. Algorithm~\ref{alg:greedy}  does not yield a monotonically decreasing profile of variances. Instead the variance of the entries selected by the algorithm switches between small and large values as the number of compromised observations increases.  This suggests, that certain entries are significantly more sensitive to additive attack than others and the existence of more vulnerable sensors that are determined by the topology of the system, as shown in~\eqref{eq:sigma_opt}.
For both cases, the probability of false alarm exhibits non-monotonic behavior with the number of compromised observations, and interestingly, the change in monotonicity coincides with the threshold.
\section{Conclusion}\label{conclusion}

We have proposed an information theoretic sparse attack construction within a  smart grid Bayesian state estimation framework. The proposed attack construction minimizes the mutual information between the state variables in the smart grid and the observations obtained by the operator while minimizing the probability of detection. To that end, we have proposed a cost function that combines the mutual information and the KL divergence that is amenable to sparse attack constructions. We have theoretically characterized the single observation attack case by proving that the resulting cost function is convex and obtaining the optimal attack construction for this case. We distill the insight obtained from the single observation case to propose a sparse attack construction via a greedy algorithm that overcomes the combinatorial challenge posed by the observation selection problem. We have numerically assessed the performance of the proposed attack in the IEEE 30 Bus Test Case and observed that the probability of detection exhibits a threshold effect when a critical number of observations are compromised.

\vspace{-5mm}
\bibliographystyle{IEEEtran}
\bibliography{IT_sparse_attacks_new}

\end{document}